\documentclass{article}
\usepackage[utf8]{inputenc}
\usepackage{url}
\usepackage{hyperref}
\usepackage[utf8]{inputenc}
\usepackage{afterpage}
\usepackage[english]{babel}
\usepackage{float}
\usepackage[T1]{fontenc}
\usepackage{graphicx, subfig}
\graphicspath{{}}
\usepackage{fancyhdr}
\usepackage{lmodern}
\usepackage{color}
\usepackage{bbm}
\usepackage{amsfonts}
\usepackage{amsmath}
\usepackage{amsthm}
\usepackage{amssymb}
\usepackage{blindtext}
\usepackage{relsize}
\usepackage{breqn}
\usepackage{mathtools}
\usepackage{bm}
\usepackage[
backend=biber,
style=alphabetic,
]{biblatex}
\usepackage[a4paper, total={6in, 10in}]{geometry}
\usepackage{authblk}

\newcommand{\R}{\mathbb{R}}
\newcommand{\C}{\mathbb{C}}

\newcommand{\A}{\mathcal{A}}
\newcommand{\CN}{\mathcal{CN}}
\newcommand{\As}{\mathcal{A}^*}

\newcommand{\frob}[1]{ \left\| #1 \right\|_{F} }
\newcommand{\nucl}[1]{ \left\| #1 \right\|_{*}}
\newcommand{\twonorm}[1]{\left\| #1 \right\|_{2}}
\newcommand{\spec}[1]{\left\| #1 \right\|}
\newcommand{\scal}[1]{\left\langle #1 \right\rangle }

\newcommand{\cone}[1]{\mathcal{K}_{*}( #1 )}
\newcommand{\tr}{\text{tr \,}}

\newcommand{\Tr}{\text{Tr}\,}
\newcommand{\tangent}[1]{\mathcal{T}_{#1}}
\newcommand{\tangentperp}[1]{\mathcal{T}^{\perp}_{#1}}
\newcommand{\PT}[1]{\mathcal{P}_{\tangent{#1}}}
\newcommand{\PTperp}[1]{\mathcal{P}_{\tangentperp{#1}}}

\numberwithin{equation}{section}

\theoremstyle{theorem}
\newtheorem{theorem}{Theorem}[section]

\theoremstyle{theorem}
\newtheorem{lemma}[theorem]{Lemma}

\theoremstyle{definition}
\newtheorem{definition}{Definition}[section]

\theoremstyle{theorem}
\newtheorem{proposition}{Proposition}[section]

\newtheorem*{remark}{Remark}

\title{How robust is randomized blind deconvolution via nuclear norm minimization against adversarial noise?}

\author[1,2]{Julia Kostin}
\author[1,2,3]{Felix Krahmer}
\author[4]{Dominik Stöger}
\affil[1]{Technical University of Munich, Department of Mathematics}
\affil[2]{Munich Center for Machine Learning}
\affil[3]{Technical University of Munich, Munich Data Science Institute}
\affil[4]{KU Eichst\"att-Ingolstadt, Mathematical Institute for Machine Learning and Data Science (MIDS)}

\date{March 16, 2022}

\begin{document}

\maketitle

\begin{abstract}

In this paper, we study the problem of recovering two unknown signals from their convolution, which is commonly referred to as blind deconvolution. Reformulation of blind deconvolution as a low-rank recovery problem has led to multiple theoretical recovery guarantees in the past decade due to the success of the nuclear norm minimization heuristic. In particular, in the absence of noise, exact recovery has been established for sufficiently incoherent signals contained in lower-dimensional subspaces. However, if the convolution is corrupted by additive bounded noise, the stability of the recovery problem remains much less understood. In particular, existing reconstruction bounds involve large dimension factors and therefore fail to explain the empirical evidence for dimension-independent robustness of nuclear norm minimization. Recently, theoretical evidence has emerged for ill-posed behavior of low-rank matrix recovery for sufficiently small noise levels. In this work, we develop improved recovery guarantees for blind deconvolution with adversarial noise which exhibit square-root scaling in the noise level. Hence, our results are consistent with existing counterexamples which speak against linear scaling in the noise level as demonstrated for related low-rank matrix recovery problems. 

\end{abstract}


\section{Introduction}

Blind deconvolution refers to the problem of recovering two unknown signals $x$ and $w$ from their convolution
\begin{equation}\label{equ:circconv}
    y= w \ast x.
\end{equation}
This problem and its generalization, blind demixing, arise from various applications in signal processing, including wireless communications \cite{wirelesscomm1998,wirelesscomm2019} and imaging \cite{imaging1996,imaging2009,imaging2017}. 
For arbitrary pairs of signals $(w, x)$ the convolution map is not invertible and thus, the problem is inherently ill-posed. 
In this paper, we focus on the problem of randomized blind deconvolution as introduced in \cite{ahmedBD}, which appears, e.g., in the context of channel estimation.
As in \cite{ahmedBD}, we assume that convolution the in \eqref{equ:circconv} is circular and that $w$ and $x$ are contained in some known subspaces.
\\[5pt]
One approach for solving the randomized blind deconvolution problem, introduced by \cite{ahmedBD}, consists of recasting this bilinear inverse problem as a linear rank-1 matrix recovery problem. 
Indeed, since the signals can be only recovered up to a scaling factor, reconstructing $w$ and $x$ is equivalent to recovering their rank-$1$ outer product. \\[5pt]
Following \cite{ahmedBD}, we assume in the following that $x, w \in \C^L$ are contained in some known lower-dimensional subspaces

\begin{align*}
    w &= B h_0, \quad h_0 \in \C^{K}, \\
    x &= C m_0, \quad m_0 \in \C^{N},
\end{align*}
where $B \in \C^{L \times K}$ is a deterministic matrix and $C \in \C^{L \times N}$ is a random matrix. Then the blind deconvolution problem can be reformulated as the problem of recovering the rank-$1$ matrix $h_0 m_0^*$ from a set of $L$ measurements $\A(h_0 m_0^*)$ described by a linear operator $\A: \C^{K \times N} \to \C^L$, see Section \ref{background:BlindDeconvolution} for details.
\\[5pt]
A natural approach to recover the rank-1 ground truth is to find a lowest-rank matrix which results in the given measurements.
However, since this is an NP-hard problem \cite{nphard} one often considers the nuclear norm minimization approach instead, a convex heuristic \cite{RechtFazelGuaranteed2010}, where one instead aims to minimize the nuclear norm of the matrix subject to the measurement constraints:
\begin{equation}\label{noiselessSDP}
\begin{split}
    &\text{minimize } \nucl{X} \\
    &\text{subject to } \A(X) = y.
\end{split}
\end{equation}
In the absence of noise, \cite{ahmedBD} established exact recovery of the ground truth with high probability given a near-optimal number of measurements $\frac{L}{\log^3 L} \gtrsim \max \left( \mu^2_{\max} K, \mu^2_h N \right)$, where $\mu_{\max}$ and $\mu_{h}$ are coherence parameters, describing roughly the spread of rows of $B$ in the Fourier space and the alignment of $h_0$ with respect to $B$ respectively (see Section \ref{background:BlindDeconvolution}). 
\\[5pt]
However, establishing near-optimal recovery guarantees becomes much more challenging when the measurements are corrupted with noise, i.e., the observations are given by
\begin{equation*}
y = \A(h_0 m_0^*) + e,
\end{equation*}
where the entries of the noise vector $e \in \C^L$ are often either assumed to be random (such as i.i.d. Gaussian \cite{LiStrohmerRapidDBNonconvex2016,HuangHandBDRiemannian2018} or sub-Gaussian \cite{chenBDminimaxoptimal2021}) or  the noise vector $e$ is assumed to be deterministic and bounded with respect to the  $\ell_2$-norm \cite{ahmedBD,strohmer2017demixing,jung2017optimal}.
The latter is sometimes referred to the adversarial noise scenario, as it allows for noise specifically designed to be produce maximal possible reconstruction error.
\\[5pt]
In the case where the noise vector is deterministic and bounded in $\ell_2$-norm, i.e., satisfies $\twonorm{e} \leq \tau$ for some noise level $\tau > 0$, noisy blind deconvolution can be tackled via the constrained nuclear minimization program 
\begin{equation}\label{noisySDPintro}
\begin{split}
    &\text{minimize } \nucl{X} \\
    &\text{subject to } \twonorm{\A(X) - y} \leq \tau.
\end{split}
\end{equation}
Despite the exact recovery guarantees for noiseless measurements, robustness of blind deconvolution via the program \eqref{noisySDPintro} remains much less understood. In particular, existing recovery guarantees \cite{ahmedBD,strohmer2017demixing,jung2017optimal}, which rely on techniques based on dual certificates (see Section \ref{background:DCandtheGolfingScheme}), take the form 
\begin{equation}\label{prevGuarantees}
    \frob{X^* - X_0} \lesssim \sqrt{\min\left(K, N \right)} \tau,
\end{equation}
where $X^*$ denotes a minimizer of \eqref{noisySDPintro}. 
Crucially, in \eqref{prevGuarantees} an additional multiplicative factor $\sqrt{\min\left(K, N \right)} $ appears.
This is in stark contrast to comparable results for low-rank matrix recovery problems involving Gaussian measurement ensembles \cite{CandesTightOracle} and Phase Retrieval \cite{candesli_2014,kueng2014phase}, where this additional dimension factor does not appear. 
\\[5pt]
This raises the natural question whether this additional dimension factor can be removed.
Recently, \cite{krahmer2020} have shown that this additional dimension factor is indeed not an artifact of the proof. 
Namely, there exist incoherent isometric embedding maps $B$, such that for sufficiently small noise levels, the problem admits an alternative solution $\tilde{X}$ which is feasible, preferred and far from the true solution in the sense that
\begin{equation*}
    \frob{\tilde{X}-X_0} \gtrsim \tau \sqrt{\frac{KN}{L}},
\end{equation*}
see Theorem \ref{instability} below.
\\[5pt]
Utilizing the fact that the unstable behaviour only arises for very small noise levels, \cite{krahmer2020} provide a near-optimal error bound for blind deconvolution, which, however, only holds for sufficiently large noise levels. Hence, it remains unclear whether a near-optimal recovery bound holds uniformly for all noise levels and whether the existing guarantees can be improved. 
\\[5pt]
In this article, we make a step towards understanding the noise-dependent robustness of blind deconvolution in the adversarial setting by providing improved recovery guarantees which continuously depend on the noise level. For small noise levels, our result allows for quadratic scaling of the reconstruction error which is consistent with the unstable behaviour demonstrated in \cite{krahmer2020}. For larger noise levels, we provide linear scaling in the noise level given a near-optimal number of measurements. \\
\textbf{Theorem 3.1 (informal).} \textit{If the number of measurements $L$ exceeds $\mathcal{O}(\left( K + N \right) \log^3 L)$, then it holds with high probability for any minimizer $X^*$ of the convex program \eqref{noisySDPintro} that }
\begin{equation*}
    \frob{X^* - X_0} \lesssim \max \left\{ (\log( L))^{1/4} \sqrt{\frob{X_0} \tau }, \sqrt{\log( L)} \tau \right\}.
\end{equation*}

In our analysis of the problem, we combine refined descent cone analysis of the nuclear norm (see Section \ref{background:DescentConeAnalysis}) with the classical proof methods based on so-called approximate dual certificates introduced in \cite{noisyMC,gross11} for matrix completion (see also Section \ref{background:DCandtheGolfingScheme}).
\\[5pt]
The dimension of the lifted nuclear norm minimization problem greatly exceeds the number of degrees of freedom of the original problem, and thus solving \eqref{noisySDPintro} in practice is computationally demanding. Recently, memory-efficient algorithms have been suggested to speed up solution of constrained semidefinite programs \cite{troppSDP}. 
On the other hand, blind deconvolution and related low-rank matrix recovery problems have been successfully tackled via nonconvex optimization \cite{LiStrohmerRapidDBNonconvex2016,HuangHandBDRiemannian2018,chenBDminimaxoptimal2021}, which is faster and more computationally efficient. \cite{LiStrohmerRapidDBNonconvex2016} and \cite{HuangHandBDRiemannian2018} established robustness guarantees for blind deconvolution via regularized Wirtinger gradient descent with spectral initialization for a near-optimal number of measurements in the case that the noise $e$ is i.i.d. complex Gaussian. \cite{chenBDminimaxoptimal2021} significantly improved existing statistical guarantees for blind deconvolution under sub-Gaussian noise by utilizing the observation that the solutions of the convex and nonconvex problems are very close. This allowed the authors to utilize nonconvex techniques to establish robustness of the convex program: for a noise vector with entries obeying $\spec{e_i}_{\psi_2} \leq \sigma$,\footnote{Here, $\spec{\cdot}_{\psi_2}$ is the sub-Gaussian norm, see, e.g., \cite{vershynin_2018}.} and a number of measurements $L \gtrsim \mu^2 \max \{K, N\} \log^9 L$, it holds with high probability that
\begin{equation}\label{chenminimaxbound}
    \spec{X^* - X_0} \lesssim \sqrt{K \log L} \sigma.
\end{equation}
This improves previous recovery bounds for convex relaxation such as \cite{ahmedBD} by a factor $\sqrt{\frac{L}{\log L}}$. However, the result assumes an upper bound on the noise level $\sigma \leq c \frac{\frob{h_0 m_0^*}}{\sqrt{K \log^5 L}}$. Furthermore, establishing closeness of the approximate nonconvex and the convex solutions in \cite{chenBDminimaxoptimal2021} requires sub-Gaussianity of the noise entries. Therefore, the employed proof methods do not readily extend to the adversarial noise setting. Since the evidence in \cite{krahmer2020} suggests significant instability of blind deconvolution with adversarial noise, recovery guarantees of form \eqref{chenminimaxbound} cannot hold with high probability, and thus, a different analysis is required for the adversarial setting. The proof techniques presented in this work do not rely on randomness of the noise, and hence, they could be potentially interesting for establishing adversarial robustness of both convex and nonconvex algorithms for blind deconvolution and related low-rank matrix recovery problems such as matrix completion and phase retrieval.

\subsection{Organisation of the paper and our contribution}
In Section \ref{background}, we review the formulation of blind deconvolution as a low-rank matrix recovery problem and summarize previous results establishing exact and robust recovery for blind deconvolution via convex programming as well as their key proof techniques.  
    In Section \ref{mainresult}, we present our main result, a noise-dependent reconstruction error bound for blind deconvolution via nuclear norm minimization, and formulate the key lemmas necessary for its proof. Section \ref{proofoftheauxiliaryresults} provides proofs of the these lemmas. We discuss implications of our findings and remaining open questions in Section \ref{outlook}. 
\\[5pt]
Our result contributes to the understanding of (in)stability phenomena in blind deconvolution, and more generally, low-rank matrix recovery problems corrupted with adversarial noise. Our proof is based on a novel, more refined descent cone analysis of rank-$1$ matrices.
To the best of our knowledge, this is the first result establishing recovery guarantees for arbitrary noise levels without amplification factors which scale polynomially in the dimension of the problem.

\section{Background and related work}\label{background}
\subsection{Blind Deconvolution}\label{background:BlindDeconvolution}
In this section, we will summarize the reformulation of blind deconvolution as a low-rank matrix recovery problem, including the subspace constraints on the ground truth signals and the incoherence assumptions. 
\\[5pt]
We want to recover two unknown signals $w \in \C^L$ and $x \in \C^L$ from their circular convolution 
\begin{equation*}
    w \circledast x := \left( \sum_{j=1}^{L} w_j x_{k-j} \right)_{k=1}^{L}.
\end{equation*}
We assume that $w$ is constrained in a $K$-dimensional subspace of $\C^L$ defined by a deterministic isometry $B \in \C^{L \times K}$, $B^* B = Id_K$:
\begin{equation*}
    w = B h_0, \quad h_0 \in \C^{K} \setminus \{ 0 \}.
\end{equation*}
This assumption was originally motivated by applications in wireless communications \cite{ahmedBD}, where $B$ consists of a subset of columns of the identity matrix. 
The second signal $x \in \C^L$ is assumed to be an element of an $N$-dimensional random subspace generated by a matrix $C \in \C^{L \times N}$ with i.i.d. complex Gaussian entries $C_{ij} \sim \CN(0,\frac{1}{\sqrt{L}}) $:
\begin{equation*}
    x = C \overline{m}_0, \quad m_0 \in \C^{N} \setminus \{ 0 \},
\end{equation*}
where the complex conjugate is used for a more convenient presentation below. This assumption can be seen as encoding the original message $m_0$ using random waveforms \cite{randomchannelcoding2009}.  \\[5pt]
Although the convolution map is bilinear in $w$ and $x$, the discrete Fourier transform of $w \circledast x$ can be rewritten as a map linear in the outer product $h_0 m_0^*$ of the original signals. Namely, by the convolution theorem we have
\begin{equation}\label{FourierConv}
    F(w \ast x) = \sqrt{L} Fw \odot Fx = \sqrt{L} (F B h_0) \odot (F C \overline{m_0}),
\end{equation}
where $F$ denotes the normalized, unitary discrete Fourier matrix.
We denote by $b_\ell$ the $\ell$th row of the matrix $\overline{FB}$ and by $c_{\ell}$ the $\ell$th row of $\sqrt{L} F C$. With this notation, the $\ell$th entry of the Fourier transform of the convolution can be written as 
\begin{equation}\label{FourierToLinear}
    F(w \ast x)_\ell = \scal{b_\ell, h_0} \scal{m_0, c_\ell} = b_{\ell}^* h_0 m_0^* c_{\ell} = \tr (h_0 m_0^* c_{\ell} b_{\ell}^*) = \scal{b_{\ell} c_{\ell}^*, h_0 m_0^* }_F.
\end{equation}
Thus, observing the circular convolution of $w$ and $x$ is equivalent to observing the set of linear rank-$1$ measurements $\left( \scal{b_{\ell} c_{\ell}^*, h_0 m_0^* }_F \right)_{\ell=1}^{L}$. Since $w$ and $x$ can only be recovered up to an inherent scaling ambiguity, recovering $(w, x)$ is equivalent to reconstructing $X_0 \coloneqq h_0 m_0^*$. 
Thus, we define the measurement operator as introduced in \cite{ahmedBD}:
\begin{align}\label{LinearOperator}
\begin{split}
    &\A:\, \C^{K \times N} \longrightarrow \C^{L} \\   
    &(\A(X))_\ell \coloneqq \scal{b_\ell c_\ell^*, X}_F.
\end{split}
\end{align}
The following two quantities have been shown in \cite{ahmedBD} to be crucial for the success of the reconstruction of the ground truth signals:
\begin{align*}
    &\mu^2_{\max} \coloneqq \frac{L}{K} \max_{\ell \in [L]} \twonorm{b_\ell}^2, \\
    &\mu^2_{h_0} \coloneqq \frac{L}{\twonorm{h_0}^2} \max_{\ell \in [L]} \left| \scal{b_{\ell}, h_0}\right|^2.
\end{align*}
It follows from these definitions that $1 \leq \mu^2_{\max} \leq L/K$ and $1 \leq \mu^2_{h_0} \leq K \mu^2_{\max}$. 

\subsection{Nuclear norm minimization}\label{background:nuclearnormminimization}

A natural approach to recover $X_0$ is to search for a minimum-rank matrix $X^*$ satisfying $\A(X^*) = \A(X_0)$. However, rank minimization is NP-hard \cite{nphard} and thus not computationally tractable. A common convex heuristic for low-rank matrix recovery problems, introduced by \cite{fazel2002}, is minimizing the nuclear norm of a matrix (i.e., the sum of its singular values) subject to a set of measurement constraints:
\begin{equation}\label{noiselessSDP}
\begin{split}
    &\text{minimize } \nucl{X}  \\
    &\text{subject to } \A(X) = y.
\end{split}
\end{equation}
In the absence of noise, the semidefinite program \eqref{noiselessSDP} recovers $X_0$ exactly with high probability given a near-optimal number of observations.
We include the more recent version of this result, which originally was derived in the context of the more general framework of blind demixing \cite{strohmer2017demixing,jung2017optimal} – reconstruction of multiple pairs of signals $\left( (w_i, x_i) \right)_{i=1}^{r}$ from a sum of their circular convolutions.

\begin{theorem}[
\cite{jung2017optimal}]\label{NoiselessRecovery}
    Let $\omega \geq 1$, $h_0 \in \C^K$, $m_0 \in \C^N$ and $y = \A(h_0 m_0^*) \in \C^L$. Assume that
    \begin{equation*}
        L \geq C \omega \left( K \mu_{\max}^2 \log(K \mu_{\max}^2) + N \mu_{h_0,\omega}^2 \right) \log^3 L,
    \end{equation*}
    where $C > 0$ is an absolute constant and $\mu_{h_0,\omega}$ is a coherence parameter arising from the Golfing scheme (see Section \ref{appendix:golfingscheme}, Definition \eqref{muhomegadefinition}).  Then, with probability $1 - \mathcal{O}(L^{-\omega})$, the matrix $X_0 = h_0 m_0^*$ is the unique minimizer of the convex program \eqref{noiselessSDP}.
\end{theorem}
For generic $h_0$ and $B$, $\mu_{\max}$ and $\mu_{h_0,\omega}$ are reasonably small, i.e., at the order of small constant. In this case, the number of samples needed for exact recovery in Theorem \ref{NoiselessRecovery} is near-optimal up to logarithmic factors, since the number of degrees of freedom of the problem at the order of $K+N$.
\\[5pt]
Now, we consider recovery of $X_0$ in the case where the measurements are corrupted by bounded, possibly adversarial noise:
\begin{equation*}
    y = \A(X_0) + e, \; \twonorm{e} \leq \tau, \; \tau \in \R_{+}.
\end{equation*}
The convex recovery program becomes 
\begin{equation}\label{noisySDP}
\begin{split}
    &\text{minimize } \nucl{X}  \\
    &\text{subject to } \twonorm{\A(X)-y} \leq \tau.
\end{split}
\end{equation}
In the presence of noise, exact recovery of $X_0$ is in general no longer possible. However, it can be established that any minimizer of \eqref{noisySDP} is relatively close to the ground truth:
\begin{theorem}[\cite{ahmedBD,jung2017optimal}]\label{NoisyRecovery}
    Given observations $y = \A(X_0) + e \in \C^L$ where $\twonorm{e} \leq \tau$ and under the same conditions as in Theorem \ref{NoiselessRecovery}, with probability at least $1 - \mathcal{O}(L^{-\omega})$ any minimizer $X^*$ of the convex program \eqref{noisySDP} satisfies 
    \begin{equation}\label{BDoldbound}
        \frob{X^*-X_0} \lesssim \tau \sqrt{ \max \left\{ 1; \frac{K \mu^2_{\max} N}{L} \right\} \log L }.
    \end{equation}
\end{theorem}
Thus, ignoring logarithmic factors, existing recovery guarantees amplify the output noise level $\tau$ by at least the factor $\sqrt{\min \{K, N \}}$ for $ L \asymp K \mu^2_{\max} + N \mu_{h_0,\omega}^2$. 
\\[5pt]
This is in stark contrast to results for other low-rank matrix recovery problems. For instance, \cite{ChandrasekaranConvexGeometry} derive dimension-independent reconstruction error bounds for low-rank matrix recovery if the measurement matrices have i.i.d. Gaussian entries:
\begin{theorem}[\cite{ChandrasekaranConvexGeometry}]\label{robustnessforgaussianmeasurements}
    Let $\A: \C^{K \times N} \to \C^L$ be a random map such that $\A(X) = \sum_{j=1}^m \scal{A_j, X}_F$, where the entries of $A_j$ are i.i.d zero-mean Gaussian with variance $1/L$.
    Let $L \gtrsim r(2 \max \{K, N\} - r)$, where $r$ is the rank of the ground truth $X_0 \in \C^{K \times N}$. Let $X^*$ be a solution of the constrained nuclear norm minimization problem. Then with high probability it holds that  
    \begin{equation*}
        \frob{X_0 - X^*} \lesssim \tau.
    \end{equation*}
\end{theorem}
However, in the blind deconvolution scenario, the measurement matrices $b_{\ell} c_{\ell}^*$ are more structured than in the Gaussian scenario in Theorem \ref{robustnessforgaussianmeasurements}, opening the question whether similar dimension-independent robustness guarantees could be established for blind deconvolution with bounded noise. 
\\[5pt]
\cite{krahmer2020} tackled the question whether the descent cone analysis techniques (see Section \ref{background:DescentConeAnalysis}) employed in \cite{ChandrasekaranConvexGeometry} could be modified to deliver dimension-independent recovery guarantees for more structured measurements, such as blind deconvolution and matrix completion. The authors found that for sufficiently small noise levels, blind deconvolution under adversarial noise can be unstable:
\begin{theorem}[\cite{krahmer2020}]\label{instability}
    Assume that
    \begin{equation*}
        C_1 K \leq L \leq \frac{1}{C_2} KN.
    \end{equation*}
Then there exists an isometry $B \in \C^{L \times K}$ satisfying $\mu^2_{\max}=1$, such that for all $h_0 \in \C^K \setminus \{0\}$ and $m_0 \in \C^N \setminus \{0\}$ the following holds:
With probability at least $1 - \mathcal{O}\left( \exp\left( -\frac{K}{C_3 \mu^2_{h_0}} \right) \right)$, for all noise levels $\tau \leq \tau_0$ for some $\tau_0 > 0$ there exists an adversarial noise vector $e \in \C^L$, $\twonorm{e} \leq \tau$, such that the recovery program \eqref{noisySDP} admits an alternative solution $\tilde{X}$ with the following properties:
\begin{enumerate}
    \item $\tilde{X}$ is feasible: $\twonorm{\A(\tilde{X}) - y} \leq \tau$;
    \item $\tilde{X}$ is preferred to $X_0$: $\nucl{\tilde{X}} \leq \nucl{X_0}$;
    \item $\tilde{X}$ is far from the true solution in Frobenius norm: 
    \begin{equation}\label{instabilityequation}
        \frob{\tilde{X}-X_0} \geq \frac{\tau}{C_4} \sqrt{\frac{K N}{L}}.
    \end{equation}
\end{enumerate}
\end{theorem}

Theorem \ref{instability} suggests that the dimension factor in the recovery guarantees \eqref{BDoldbound} is not merely a proof artifact. 
Although $\tilde{X}$ constructed in the proof of Theorem \ref{instability} is preferred to the true solution, it is not a minimizer of the SDP \eqref{noisySDP}, see Remark 3.4. in \cite{krahmer2020}. Moreover, the instability result is not expected to hold if one picks a generic isometry $B$, thus randomizing over both $B$ and $C$ \cite{krahmer2020}. Still, Theorem \ref{instability} provides a lower bound for possible recovery guarantees which hold for all preferred solutions.
\\[5pt]
However, Theorem \ref{instability} gives no information on the noise level $\tau_0$ from which the instability is to be expected.
In particular, it does not exclude the possibility that for larger noise levels $\tau$, stronger recovery guarantees can be shown.
Indeed, \cite{krahmer2020} made a first step in this direction by showing near-optimal error scaling whenever the adversarial noise magnitude is not too small: 
\begin{theorem}[\cite{krahmer2020}]\label{StabilityLargeLevels}
    Let $\omega > 0$, $\mu \geq 1$. For
    \begin{equation*}
        L \gtrsim \frac{\mu^2}{\omega^2} (K + N) \log^2 L,
    \end{equation*}
    with probability at least $1 - \mathcal{O}\left( \exp \left( - \frac{L \omega^{4/3}}{C \log^{4/3} (e L) \mu^{4/3}}\right) \right)$ it holds for all $h_0 \in \C^L \setminus \{0\}$ with $\mu_{h_0} \leq \mu$ and all $m_0 \in \C^N \setminus \{0 \}$ that the minimizer $X^*$ of the convex program \eqref{noisySDP} satisfies
\begin{equation*}
    \frob{X^* - X_0} \lesssim \frac{\mu^{2/3} \log^{2/3} L }{\omega^{2/3}} \max \{ \tau, \omega \frob{X_0} \}.
\end{equation*}
\end{theorem}
Although Theorem \ref{StabilityLargeLevels} establishes near-dimension independent scaling in the noise, this only holds for sufficiently large noise levels $\tau > \omega \frob{X_0}$. 

\subsection{Proof methods for robust low-rank matrix recovery}
Before we present the main result, we would like to review two common proof techniques that have been used to establish recovery guarantees for blind deconvolution and related low-rank matrix recovery problems: descent cone analysis and dual certificates (for further details, see, e.g., \cite{Fuchs2021}).
 Our proof in Section \ref{proofofmaintheorem} will use an exact dual certificate (introduced in Section \ref{background:DCandtheGolfingScheme}) to conduct a refined descent cone analysis (reviewed in Section \ref{background:DescentConeAnalysis}). 
\subsubsection{Descent cone analysis}\label{background:DescentConeAnalysis}
Descent cone analysis aims to quantify the intersection of the feasible set $\left\{ X \in \C^{K \times N}: \twonorm{\A(X) - y} \leq \tau \right\}$ with the cone of all descent directions of the nuclear norm at the point $X_0$. 
To describe the connection between descent cone analysis and robustness guarantees for blind deconvolution, we start with the following definition:
\begin{definition}
Let $X_0 \in \C^{K \times N}$. The descent cone at $X_0$ is defined by
\begin{equation*}
    \cone{X_0} \coloneqq \{ Z \in \C^{K \times N}: \text{there exists } \varepsilon > 0 \text{ such that } \nucl{X_0 + \varepsilon Z} \leq \nucl{X_0} \}.
\end{equation*}
\end{definition}
Any preferred solution $\tilde{X}$ with $\nucl{\tilde{X}} \leq \nucl{X_0}$ can be written as $\tilde{X} = X_0 + Z$ for some $Z \in \cone{X_0}$. The reconstruction error is then equal to $\frob{X_0 - \tilde{X}} = \frob{Z}$. On the other hand, any feasible solution satisfies 
\begin{equation*}
    \twonorm{\A(\tilde{X}) - \A(X_0) } \leq \twonorm{\A(\tilde{X}) - y} + \twonorm{\A(X_0) - y} \leq 2 \tau,
\end{equation*}
therefore, a preferred and feasible solution must satisfy
\begin{equation*}
    \twonorm{\A(\tilde{X}) - \A(X_0) } = \twonorm{\A(Z)} \leq 2\tau.
\end{equation*}
This translates into the reconstruction bound
\begin{equation}\label{conicbound}
    \frob{\tilde{X}-X_0} = \frac{\frob{Z}}{\twonorm{\A(Z)}} \twonorm{\A(Z)}  \leq 2 \tau \frac{\frob{Z}}{\twonorm{\A(Z)}} \leq 2 \tau \left( \inf_{Z \in \cone{X_0} \setminus \{0 \} } \frac{\twonorm{\A(Z)}}{\frob{Z}} \right)^{-1}.
\end{equation}
This motivates the following definition.
\begin{definition}[Minimum conic singular value]\label{def:minimumconicsingularvalue}
Let $X_0 \in \C^{K \times N}$. The minimum conic singular value of $\A$ at $\cone{X_0}$ is defined by
\begin{equation*}
    \lambda_{\min}(\A, \cone{X_0}) \coloneqq \inf_{Z \in \cone{X_0}  \setminus \{0 \}} \frac{\twonorm{\A(Z)}}{\frob{Z}}.
\end{equation*}
\end{definition}
Inserting the definition in \eqref{conicbound}, we obtain the following relation between the minimum conic singular value and the reconstruction error.
\begin{theorem}[\cite{ChandrasekaranConvexGeometry}]\label{conicvaluetorecoveryguarantees}
    Let $X_0 \in \C^{K \times N}$ and $\A: \C^{K \times N} \to \C^L$ be a linear operator. Let the measurements be given by $y = \A(X_0) + e$ where $e \in \C^L$, $\twonorm{e} \leq \tau$. Then any minimizer $\hat{X}$ of the convex program \eqref{noisySDP} satisfies
    \begin{equation*}
        \frob{\hat{X}-X_0} \leq \frac{2\tau}{\lambda_{\min}(\A, \cone{X_0})}.
    \end{equation*}
\end{theorem}
In the absence of noise, exact recovery of a matrix $X_0 \in \C^{K \times N}$ is equivalent to the condition $\ker \A \cap \cone{X_0} = \{ 0 \}$ \cite{ChandrasekaranConvexGeometry}. However, for the case of blind deconvolution, in\cite[Proposition 3.3]{krahmer2020} it has been shown that in the scenario described in Theorem \ref{instability}, with high probability it holds that 
\begin{equation*}
  \lambda_{\min}(\A, \cone{X_0}) \lesssim   \sqrt{ \frac{L}{KN}}.
\end{equation*}
Therefore, we cannot utilize Theorem \ref{conicvaluetorecoveryguarantees} to provide dimension-independent reconstruction error bounds, and we need to conduct a much more refined analysis. 
A first refined analysis was presented in \cite{krahmer2020} (see Theorem \ref{StabilityLargeLevels}).
The key insight was that this bad conditioning does not hold for all descent directions, but solely for directions along which only small decrements are allowed. However, excluding these directions leads to a reconstruction error bound only valid for large noise levels. In our proof, we build upon this idea to provide a noise-level-dependent reconstruction error bound.
 The key insight will be a lower bound on $\twonorm{\A (Z) }$ in terms of an actual descent parameter, see Lemma \ref{lowerboundforbeta} below.

\subsubsection{Dual certificate}\label{background:DCandtheGolfingScheme}

Convex duality techniques, introduced in \cite{candes2006uncertainty}, have been employed to prove exact and robust recovery guarantees for compressed sensing \cite{candes2006uncertainty}, matrix completion \cite{RechtSimpler, gross11}, and  blind deconvolution \cite{ahmedBD}. According to classical convex optimization theory (see, e.g., Proposition 5.4.7 in \cite{bertsekas2003convex}), $X^*$ is a minimizer of the SDP \eqref{noiselessSDP} if and only if there exists a matrix $Y \in \partial \nucl{\cdot}(X^*)$ in the subdifferential of the nuclear norm, such that $-Y$ belongs to the normal cone $N_S(X^*)$, where $S = X_0 + \ker \A$ is the set of all feasible solutions. Since $S$ is an affine subspace, it is sufficient to find $Y \in \delta \nucl{X^*}$ orthogonal to $\ker \A$, i.e. $Y \in \text{Range\,}\As$, or $Y = \As(z)$ for some $z \in \C^L$. We will call such pair $(z,Y)$ an \emph{exact dual certificate} \cite{CandesRechtExact2009}. 
\\[5pt]
The subdifferential of the nuclear norm at a point $X$ can be characterized in terms of the singular value decomposition of $X$. In the following, we  will denote the SVD of a matrix $X$ by $X = U \Sigma V^*$, where $U \in \C^{K \times r}$ and $V \in \C^{N \times r}$ are unitary matrices and $\Sigma = \text{diag}\,(\sigma_1,...,\sigma_r)$, where $(\sigma_1,...,\sigma_r)$ are the singular values of $X$ in decreasing order. If $X$ has rank $r$, one can define the \emph{tangent space of the variety of rank-$r$ matrices} at $X$ by
\begin{equation*}
    \tangent{X} \coloneqq \left\{ U A^* + B V^*: \; A \in \C^{N \times r}, \, B \in \C^{K \times r} \right\}.
\end{equation*}
In the following, we will write $\tangent{}$ instead of $\tangent{X}$ if the base point $X$ is clear. 
\\[5pt]
The subdifferential of the nuclear norm at a point $X$ can be characterized (see, e.g., \cite{watson_subdifferential}) by 
\begin{equation*}
    \partial \nucl{\cdot}(X) = \left\{ W \in \C^{K \times N}: \: \PT{} W = UV^*, \, \spec{\PTperp{} W} \, \leq 1  \right\},
\end{equation*}
where $\PT{}$ denotes the orthogonal projection onto $\tangent{X}$, $\PTperp{}$ denotes the projection onto its orthogonal complement $\tangentperp{X}$, and $\spec{\cdot}$ is the spectral norm.
\\[5pt]
Now let $X_0 = \nu h_0  m_0^*$ be the singular value decomposition of the rank-1 matrix $X_0$, where $\nu \geq 0$ is the nuclear norm of $X_0$. According to the considerations above, the ground truth $X_0$ is a minimizer of the SDP \eqref{noiselessSDP} if and only if one can find an exact dual certificate pair $(z, Y)$ with $Y = \As(z)$ satisfying
\begin{align}
     &\PT{X_0} Y = h_0 m_0^*,\label{property1exact} \\
     &\spec{\PTperp{X_0} Y} < 1.\label{property2exact}
\end{align}
Such an object was constructed, e.g., in \cite{CandesRechtExact2009} to show exact recovery in the noiseless case for the related problem of matrix completion.
\\[5pt]
In \cite{ahmedBD}, it has been shown via the construction of the so-called \textit{approximate dual certificate} (see also Section \ref{appendix:dualcertificate}) that in the absence of noise, $X_0$ is the unique minimizer of the nuclear norm minimization problem \eqref{noiselessSDP} with high probability. Thus, the existence of an exact dual certificate follows. In the following result we state that, given a sufficient number of measurements, there exists an instance of an exact dual certificate which is bounded in the norm. 
\begin{proposition}[Exact dual certificate]\label{exactDCexistence}
    Let $\omega \geq 1$. Assume that the number of measurements satisfies
    \begin{equation}\label{LnumforexistenceofDC}
        L \geq C \omega \left( K \mu_{\max}^2 \log(K \mu_{\max}^2) + N \mu_{h_0,\omega}^2 \right) \log^3 L,
    \end{equation}
    where $C > 0$ is a universal constant and $\mu_{h_0,\omega}$ is a technical coherence parameter defined in \ref{muhomegadefinition}.
    Then with probability at least $1 - \mathcal{O}(L^{-\omega})$, there exists an exact dual certificate pair $(z,Y)$ with $Y = \As(z)$ such that
\begin{equation*}
     \PT{X_0} Y = h_0 m_0^*, \quad \spec{\PTperp{X_0} Y} \leq \frac{3}{4}.
\end{equation*}
and $z$ satisfies
\begin{equation*}
    \twonorm{z} \lesssim \sqrt{\omega \log L}.
\end{equation*}
\end{proposition}
\begin{remark}
The definition of a dual certificate only requires that $\spec{\PTperp{X_0} Y} < 1$.
However, from the constructive proof of Proposition \ref{puttingproposition}, the stronger property $\spec{\PTperp{X_0} Y} \leq \frac{3}{4}$ follows. 
Additionally, we see that Proposition \ref{exactDCexistence} also yields the existence of a vector $z \in \C^L$ with $\twonorm{z} \lesssim \sqrt{\omega \log L} $. Both of these properties will be crucial in our proof.
\end{remark}
The proof of Proposition \ref{exactDCexistence} is a combination of the construction of an approximate dual certificate via the Golfing scheme as presented in \cite{jung2017optimal} for blind demixing, as well as the construction of an exact dual certificate from the approximate one via the Putting proposition \cite{Fuchs2021}. 
For completeness, we include the proof of Proposition \ref{exactDCexistence} in the Appendix (see Sections \ref{appendix:dualcertificate}, \ref{appendix:golfingscheme} for the construction of the approximate dual certificate and the norm estimates, and Section \ref{appendix:puttingproposition} for the exact dual certificate).

\subsection{Notation}\label{notation}
Before we proceed with the formulation and the proof of the main result, we would like to introduce some common notation used in our paper. $A^*$ will denote the adjoint of a matrix $A \in \C^{K \times N}$. $\overline{a}$ will denote the complex conjugate of $a \in \C$. $\text{Re}\,a$ and $\text{Im}\,a$ will denote the real and the imaginary part of $a$, respectively. By $\log$ we will denote the natural logarithm with base $e$. $\nucl{A}$ will denote the nuclear norm of a matrix $A \in \C^{K \times N}$. It is equal to the sum of the singular values of the matrix $\sum_{i=1}^{r} \sigma_i$, where $(\sigma_1, ... , \sigma_r)$ are the singular values of $A$. By $\frob{A}$ we will denote the Frobenius norm, and by $\spec{A}$ the spectral norm of a matrix $A$. By $\twonorm{v}$ we will denote the Euclidean norm of a vector $v$. $\scal{v,w} = v^* w$ will denote the Euclidean scalar product between two vectors $v$ and $w$. $\scal{X,Z}_F = \Tr(X^* Z)$ will denote the Frobenius scalar product between two matrices. $v \odot w$ will denote the Hadamard (i.e. elementwise) product of two vectors. $\text{diag}\,v$ will denote the diagonal matrix whose diagonal is given by $v$. 
\\[5pt]
By $[n]$, we will understand the index set $\{ 1, ... , n \}$ for $n \in \mathbb{N}$. We will say $a \lesssim b$ if there exists a universal constant $C > 0$ such that $a \leq C b$. We will write $a \sim b$ if $a \lesssim b$ and $b \lesssim a$. $\text{Id}_d$ will denote the identity map on $\mathbb{R}^d$.
\\[5pt]
We will denote the cardinality of a set $S$ by $|S|$. 

\section{Main result}\label{mainresult}
Our main result establishes a noise level-dependent robustness guarantee for the blind deconvolution model via nuclear norm minimization. 
\begin{theorem}\label{mainresulttheorem}
    Let $\omega \geq 1$. Let $\A: \C^{K \times N} \to \C^L$ be given by \eqref{LinearOperator},  $X_0 = h_0 m_0^*$, where $h_0 \in \C^K$, $m_0 \in \C^N$. Let $y \in \C^L$ be given by $y = \A(X_0) + e$, where $e \in \C^L$ satisfies $\twonorm{e} \leq \tau$. Assume that the number of measurements satisfies
    \begin{equation*}
        L \geq C \omega \left( K \mu_{\max}^2 \log(K \mu_{\max}^2) + N \mu_{h_0,\omega}^2 \right) \log^3 L,
    \end{equation*}
   where $\mu_{h_0,\omega}$ is an incoherence parameter defined in \eqref{muhomegadefinition}.
    Then, with probability at least $1 - \mathcal{O}(L^{-\omega})$, it holds for any minimizer $X^*$ of the semidefinite program \eqref{noisySDP} that 
    \begin{equation}\label{mainresultequation}
\frob{X^* - X_0} \lesssim \max \left\{ (\log(\omega L))^{1/4} \sqrt{\frob{X_0} \tau }, \sqrt{\log(\omega L)} \tau \right\}.
    \end{equation}
    Here, $C > 0$ is a universal constant. 
\end{theorem}

A few comments regarding Theorem \ref{mainresulttheorem} are in order.
    Note that we are interested in the noise level regime where $ \tau \leq \frob{X_0} \approx \twonorm{\A (X_0)}$. Only in this regime we can expect nontrivial reconstruction guarantees, since, if $\tau \geq \frob{X_0}$, one could choose a noise vector $e$ such that the trivial zero solution is feasible. 
    In particular, in the relevant regime $ \tau \leq \frob{X_0} $, the bound \eqref{mainresultequation} becomes
    \begin{equation}\label{equ:errorbound}
        \frob{X^* - X_0} \lesssim  (\log(\omega L))^{1/4} \sqrt{\frob{X_0} \tau }.
    \end{equation}
    To see how this compares to the the existing dimension-dependent recovery guarantee \eqref{BDoldbound} in \cite{ahmedBD,strohmer2017demixing,jung2017optimal} we first reformulate \eqref{equ:errorbound} as
    \begin{equation*}
     \frob{X^* - X_0} \lesssim \sqrt{\frac{   (\log(\omega L))^{1/2} \frob{X_0}}{\tau}} \cdot  \tau .   
    \end{equation*}
    Ignoring logarithmic factors, for noise levels $  \frac{L}{KN\mu_{\max}^2} \frob{X_0}  \ll  \tau \leq \frob{X_0} $, this significantly improves over the dimension-dependent recovery guarantee \eqref{BDoldbound}.
    \\[5pt]
Compared to the stability result in \cite{krahmer2020} (see Theorem \ref{StabilityLargeLevels}), in our result we observe a square-root dependence of the reconstruction error bound on the noise level $\tau$ for small noise levels.
In contrast, the reconstruction error bound in Theorem \ref{StabilityLargeLevels} becomes constant whenever the noise level $\tau$ is smaller than a certain threshold. 
\\[5pt]
The bound \eqref{mainresultequation} in Theorem \ref{mainresulttheorem} becomes worse when the noise level $\tau$ becomes smaller. This reflects the instability result in \cite{krahmer2020}, see Theorem \ref{instability}, which shows the existence of an alternative solution which amplifies the output error by a dimension factor. 
\subsection{Proof of Theorem \ref{mainresulttheorem}}\label{proofofmaintheorem}
Without loss of generality, we assume $\twonorm{h_0} = \twonorm{m_0} = 1$ and write $X_0$ as 
\begin{equation}\label{SVDX0}
    X_0 = \nu h_0  m_0^*,
\end{equation}
where we have $\nu = \nucl{X_0} $. 
\\[5pt]
Let $X^*$ be a minimizer of the convex program \eqref{noisySDP}. Then it holds that $\nucl{X^*} \leq \nucl{X_0}$ and hence, there exist an $\varepsilon > 0$ and a descent cone element $Z \in \cone{X_0}$ satisfying $\frob{Z} = 1$, such that $X^*$ can be written as
\begin{equation}\label{XstarX0andZ}
    X^* = X_0 + \varepsilon Z.
\end{equation}
Recall from Section \ref{background:DescentConeAnalysis} that lower bounds on the minimum conic singular value (see Definition \ref{def:minimumconicsingularvalue})
translate into recovery guarantees for $X^*$ by means of Theorem \ref{conicvaluetorecoveryguarantees}. However, as shown in \cite{krahmer2020}, the minimum conic singular value for the blind deconvolution problem is ill-behaved. This motivates us to provide a more refined analysis for $\twonorm{\A(Z)}$ which takes into account the geometry of $Z$. First, we will introduce some notation.
\\[5pt]
We can write the descent cone element $Z$ using the orthogonal decomposition
\begin{equation*}
    Z = \PT{} Z + \PTperp{}Z,
\end{equation*}
where $\tangent{} = \{h_0 m^* + h m_0^*: \; m \in \C^N, h \in \C^K \}$ and $\tangentperp{}$ denotes the orthogonal complement of $\tangent{}$. 
\\[5pt]
We denote
\begin{equation*}
    M \coloneqq \PTperp{}Z.
\end{equation*}
Then, we can write $Z$ in the following form:
\begin{equation}\label{descentconeelementdecomposition}
    Z = - \beta h_0 m_0^* + \gamma h_0 {m_0^{\perp}}^* + \eta h_0^{\perp} m_0^* + M,
\end{equation}
where $\beta, \gamma, \eta \in \R$, $h_0^{\perp} \perp h_0$, $m_0^{\perp} \perp m_0$ and $\twonorm{{m_0^{\perp}}} = \twonorm{{h_0^{\perp}}} = 1$.
\\[5pt]
For the descent cone element $Z$, the parameter $\beta$ corresponds to the "actual decrease" of the nuclear norm, whereas the orthogonal direction $M$ can only increase the nuclear norm and is smaller than $\beta$.
This fact is captured by the following lemma.
\begin{lemma}\label{MandBeta}
   Let $Z \in \cone{X_0}$ be given by \eqref{descentconeelementdecomposition}, $\beta = - \scal{Z, h_0 m_0^*}_F $, and $M = \PTperp{} Z$. Then it holds that
    \begin{equation*}
    \nucl{M} \leq \beta.
    \end{equation*}
\end{lemma}
This lemma is a consequence of the characterization of the descent cone provided in \cite{krahmer2020}.
For the proof of Lemma \ref{MandBeta} we refer to Section \ref{proofofMandBeta}.
\\[5pt]
The importance of the parameter $\beta$ becomes apparent as the most "pathological" descent directions correspond to elements $Z$ with small $\beta$, but larger tangential components $\eta$ and $\gamma$. This motivates establishing a lower bound for $\twonorm{\A(Z)}$ in terms of the "actual descent" parameter $\beta$:
\begin{lemma}[Lower bound for $\twonorm{\A(Z)}$]\label{lowerboundforAZ}
Let $\omega \geq 1$. Assume that the number of measurements satisfies 
\begin{equation*}
    L \geq C \omega \left( K \mu_{\max}^2 \log(K \mu_{\max}^2) + N \mu_{h_0,\omega}^2 \right) \log^3 L.
\end{equation*}
Then, with probability at least $1-\mathcal{O}(L^{-\omega})$, it holds for all $Z \in \cone{X_0}$ with $\frob{Z} = 1$ that

\begin{equation*}
    \twonorm{\A(Z)} \gtrsim \frac{1}{\sqrt{\log(\omega L)}} \beta,
\end{equation*}
where $\beta = -\scal{Z, h_0 m_0^*}_F$.
\end{lemma}
The proof of this lemma relies on the fact that $\twonorm{\A(Z)}$ can be lower bounded by its scalar product with the exact dual certificate, which is then decomposed into tangential and orthogonal components. For the full proof of the lemma we refer to Section \ref{proofoflowerboundforAZ}.
\\[5pt]
In the next lemma, we establish a lower bound for the "actual descent parameter" $\beta$ which depends only on the deviation from the ground truth $\varepsilon = \frob{X^* - X_0}$ and the Frobenius norm of the ground truth $\frob{X_0} = \nucl{X_0} = \nu$. 
In other words, Lemma \ref{lowerboundforbeta} states that the size of the step in the descent direction is bounded by the proportion of the descent that points in the opposite direction of the ground truth.

\begin{lemma}[Lower bound for $\beta$]\label{lowerboundforbeta}
Let $\varepsilon > 0$ and $Z \in \C^{K \times N}$ with $\frob{Z} = 1$ such that $\nucl{X_0 + \varepsilon Z} \leq \nucl{X_0}$. Let $\beta = - \scal{Z, h_0 m_0^*}_F$. Then it holds that 
\begin{equation}\label{betafinallowerbound}
    \beta \geq \min \left\{\frac{\varepsilon}{4\nu}, \frac{1}{2} \right\},
\end{equation}
where $\nu = \frob{X_0} = \nucl{X_0}$.
\end{lemma}
The proof of this result can be found in Section \ref{proofoflowerboundforbeta}.
\\[5pt]
Summarizing the results above, we have obtained a lower bound on $\beta$ which depends only on the magnitude of the reconstruction error $\varepsilon$. Combining this estimate with Lemma \ref{lowerboundforAZ}, we obtain a lower bound for $ \twonorm{\A (Z)}$ for $Z$ as in \eqref{XstarX0andZ} via
\begin{equation}\label{finalAZlowerbound}
    \twonorm{\A(Z)} \gtrsim \frac{1}{\sqrt{\log(\omega L)}} \beta \gtrsim  \frac{1}{\sqrt{\log(\omega L)}} \min \left\{\frac{\varepsilon}{4\nu}, \frac{1}{2} \right\} \gtrsim \frac{1}{\sqrt{\log(\omega L)}} \min \left\{\frac{\varepsilon}{\nu}, 1\right\}.
\end{equation}
Next, we note that
\begin{align*}
    \varepsilon \twonorm{\A(Z)} 
    &=\twonorm{\A(X^*) - \A(X_0)}\\ 
    &\le \twonorm{\A(X^*) - y} + \twonorm{y - \A(X_0)}\\
    &=  \twonorm{y - \A(X^*)} + \twonorm{e}\\
    &\leq 2 \tau,
\end{align*}
where in the last line we use that $X^*$ is feasible as well as $ \twonorm{e} \leq \tau$ by assumption. 
Combining this inequality chain with \eqref{finalAZlowerbound}, we obtain that
\begin{equation*}
    \frac{1}{\sqrt{\log(\omega L)}} \min \left\{\frac{\varepsilon}{\nu}, 1\right\} \varepsilon \lesssim \tau.
\end{equation*}
Thus, it follows that
\begin{equation*}
\varepsilon \lesssim \max \left\{ (\log(\omega L))^{1/4} \sqrt{\nu \tau }, \sqrt{\log(\omega L)} \tau \right\}.
\end{equation*}
Since $\frob{X^* - X_0} = \frob{\varepsilon Z} = \varepsilon$, we obtain the final reconstruction bound 
\begin{equation*}
     \frob{X^* - X_0} \lesssim \max \left\{ (\log(\omega L))^{1/4} \sqrt{\frob{X_0} \tau }, \sqrt{\log(\omega L)} \tau \right\}.
\end{equation*}

\section{Proof of auxiliary results}\label{proofoftheauxiliaryresults}
\subsection{Proof of Lemma \ref{MandBeta} }\label{proofofMandBeta}
For the proof we use the following characterization of the (closure of the) descent cone of the nuclear norm:
\begin{lemma}[See, e.g., \cite{krahmer2020}]\label{descentconechar}
    Let $X \in \C^{K \times N}$ be a rank-$r$ matrix with singular value decomposition $X = U \Sigma V^*$. Then
    \begin{equation*}
        \overline{\cone{X}} = \left\{Z \in \C^{K \times N}: \: - Re ( \scal{UV^*, Z}_F ) \geq \nucl{\PTperp{X}(Z)} \right\},
    \end{equation*}
where $\overline{\cone{X}}$ denotes the closure of the descent cone at $X$.
\end{lemma}
Since $Z \in \cone{X_0}$, it holds that $Z \in \overline{\cone{X_0}}$ and from Lemma \ref{descentconechar} it follows that 
\begin{equation}\label{MsmallerBeta}
    \beta = - Re ( \scal{h_0 m_0 ^*, Z}_F ) \geq \nucl{\PTperp{X}(Z)} = \nucl{M}.
\end{equation}
In particular, it follows that $\beta \geq 0$.

\subsection{Proof of Lemma \ref{lowerboundforAZ}}\label{proofoflowerboundforAZ}
Recall that by Proposition \ref{exactDCexistence}, with probability at least $1-\mathcal{O}(L^{-\omega})$ there exists an exact dual certificate pair $(z, Y)$, where
\begin{equation*}
    \twonorm{z} \lesssim \sqrt{\omega \log L}
\end{equation*}
and $Y$ satisfies properties \eqref{property1exact} and \eqref{property2exact}.
We recall that $Z \in \cone{X_0}$ and $\frob{Z} = 1$. Thus, we can estimate
\begin{equation}\label{normofAZ}
    \twonorm{\A(Z)} \geq \frac{1}{\twonorm{z}} | \scal{z, \A(Z)} | \gtrsim \frac{1}{\sqrt{\omega \log L}} | \scal{z, \A(Z)} |.
\end{equation}
We now observe that 
\begin{align*}
\scal{z, \A(Z)} = \scal{\As(z), Z} = \scal{Y, Z},
\end{align*}
where we have used that $Y = \As(z)$.
Next, we decompose the exact dual certificate $Y$ into components parallel and orthogonal to the tangent space $\tangent{X_0}$: 
\begin{align}
    \scal{Y, Z} &= \scal{\PT{X_0} Y + \PTperp{X_0} Y, Z} = \scal{\PT{X_0} Y, \PT{X_0} Z} + \scal{\PTperp{X_0} Y, \PTperp{X_0} Z} \nonumber  \\
    &=  \scal{\PT{X_0} Y - h_0 m_0^*, \PT{X_0} Z} + \scal{ h_0 m_0^*, \PT{X_0} Z} + \scal{\PTperp{X_0} Y, \PTperp{X_0} Z} \label{threeTerms}
\end{align}
where in the first line, we have used the idempotence of $\PT{X_0}$ and $\PTperp{X_0}$, and in the second line, we have added and substracted $h_0 m_0^*$. 
\\[5pt]
In \eqref{threeTerms}, the first term is equal to zero since $Y$ satisfies property \eqref{property1exact}. For the second term in \eqref{threeTerms}, we obtain
\begin{equation*}
    \scal{ h_0 m_0^*, \PT{X_0} Z} = \scal{ h_0 m_0^*, - \beta h_0 m_0^* + \gamma h_0 {m_0^{\perp}}^* + \eta h_0^{\perp} m_0^*} = - \beta,
\end{equation*}
since $m_0^{\perp} \perp m_0$ and $h_0^{\perp} \perp h_0$. For the third term in \eqref{threeTerms}, we observe using Hölder's inequality that
\begin{equation*}
    \left| \scal{\PTperp{X_0} Y, \PTperp{X_0} Z} \right| \leq \spec{\PTperp{X_0} Y} \nucl{\PTperp{X_0} Z}.
\end{equation*}
From Proposition \ref{exactDCexistence} it follows that 
\begin{equation*}
    \spec{\PTperp{X_0}Y' }  \leq \frac{3}{4}. 
\end{equation*}
Moreover, from Lemma \ref{MandBeta} we recall that
\begin{equation*}
    \nucl{M} \leq \beta.
\end{equation*}
In total, we obtain that
\begin{equation*}
    \left| \scal{\PTperp{X_0} Y, \PTperp{X_0} Z} \right| \leq \frac{3}{4} \beta.
\end{equation*}
After summation of the terms in \eqref{threeTerms}, it follows that
\begin{equation*}
    \scal{Y, Z} \leq -\frac{\beta}{4},
\end{equation*}
and thus, \eqref{normofAZ} implies
\begin{equation}\label{AZlowerboundwithbeta}
    \twonorm{\A(Z)} \gtrsim \frac{1}{\sqrt{\omega \log L}} \beta.
\end{equation}

\subsection{Proof of Lemma \ref{lowerboundforbeta}}\label{proofoflowerboundforbeta}
Our objective is to provide a lower bound on the quantity $\beta = - \scal{Z, h_0 m_0^*}_F$.
We recall that
\begin{equation}\label{ineq1}
    \nucl{X^*} = \nucl{X_0 + \varepsilon Z} \leq \nucl{X_0} = \nu.
\end{equation}
Utilizing the decomposition \eqref{descentconeelementdecomposition} of $Z \in \cone{X_0}$, we obtain for the projection of $X^*$ onto the tangent space:
\begin{equation*}
    \PT{X_0} X^* = (\nu - \varepsilon \beta) h_0 m_0^* + \varepsilon \gamma h_0 {m_0^{\perp}}^* + \varepsilon \eta h_0^{\perp} m_0^*.
\end{equation*}
$ \PT{X_0} X^*$ is (at most) a rank-$2$ matrix. $(h_0, h_0^{\perp})$ and $(m_0, m_0^{\perp})$ are orthonormal sets. We can thus write $ \PT{X_0} X^*$ in the block matrix form
\begin{equation}\label{blockmatrixdecomposition}
    \PT{X_0} X^* = 
    \begin{pmatrix}
    h_0   & h_0^{\perp} 
\end{pmatrix}
    \begin{pmatrix}
\nu - \varepsilon \beta & \varepsilon \gamma \\
\varepsilon \eta & 0 
\end{pmatrix}
\begin{pmatrix}
m^*_0 \\
 {m_0^{\perp}}^*
\end{pmatrix}.
\end{equation}
We observe that projecting $X^*$ onto the tangent space $\tangent{X_0}$ decreases its nuclear norm, i.e.,
\begin{equation}\label{ineq2}
    \nucl{X^*} \geq \nucl{\PT{X_0} X^*},
\end{equation}
since
\begin{align*}
    \nucl{X^*} &= \sup_{\spec{A} \leq 1} \scal{A, X^*} \geq \sup_{\spec{A} \leq 1, A \in \tangent{X_0}} \scal{A, X^*} \\
    &= \sup_{\spec{A} \leq 1} \scal{ \PT{X_0} A, X^*} = \sup_{\spec{A} \leq 1} \scal{ A, \PT{X_0} X^*} = \nucl{\PT{X_0} X^*}.
\end{align*}
We can compute $\nucl{\PT{X_0} X^*}$ explicitly from its matrix decomposition \eqref{blockmatrixdecomposition}: since $\begin{pmatrix}
    h_0   & h_0^{\perp} 
\end{pmatrix}$ and $\begin{pmatrix}
m^*_0 \\
 {m_0^{\perp}}^*
\end{pmatrix}$ are unitary transformations, it holds that
\begin{equation}\label{ineq3}
    \nucl{\PT{X_0} X^*} = \nucl{\begin{pmatrix}
    \nu - \varepsilon \beta &  \varepsilon \gamma \\
    \varepsilon \eta & 0
    \end{pmatrix}}.
\end{equation}
Computing the nuclear norm of this $2 \times 2$ matrix explicitly yields 
\begin{equation}\label{explicitnuclearnorm}
    \nucl{\PT{X_0} X^*} = \sqrt{(\nu - \varepsilon \beta)^2 + \varepsilon^2 ( | \gamma | + | \eta |^2)^2}.
\end{equation}
Putting the chain of (in)equalities \eqref{ineq1}, \eqref{ineq2}, \eqref{explicitnuclearnorm} together and squaring both sides results in
\begin{equation*}
    \nu^2 \geq (\nu - \varepsilon \beta)^2 + \varepsilon^2 ( | \gamma | + | \eta |^2)^2,
\end{equation*}
which is a quadratic inequality in $\beta$:
\begin{equation}\label{betainequality}
    \varepsilon^2 \beta^2 - 2 \nu \varepsilon \beta + \varepsilon^2(|\gamma| + |\eta|)^2 \leq 0.
\end{equation}
Since the problem is symmetric in the parameters $\gamma$ and $\eta$, we introduce the notation $\zeta \coloneqq |\gamma| + |\eta|$. Solving \eqref{betainequality} provides the following conditions on $\beta$ and $\zeta$:
\begin{align}
    \beta &\in \left[ \frac{\nu}{\varepsilon} -\sqrt{\frac{\nu^2}{\varepsilon^2}-\zeta^2} ; \frac{\nu}{\varepsilon} + \sqrt{\frac{\nu^2}{\varepsilon^2}-\zeta^2} \right],\label{betabound1} \\
    \zeta^2 &\leq \frac{\nu^2}{\varepsilon^2}, \nonumber
\end{align}
where we have used that $\varepsilon \geq 0$, $\zeta \geq 0$, $\nu \geq 0$. 
For the lower bound of the interval \eqref{betabound1}, the following estimate holds:
\begin{equation*}
    \frac{\nu}{\varepsilon} -\sqrt{\frac{\nu^2}{\varepsilon^2}-\zeta^2} \geq \frac{\zeta^2 \varepsilon}{2\nu}.
\end{equation*}
To see this, one can rewrite 
\begin{equation*}
    \frac{\nu}{\varepsilon} -\sqrt{\frac{\nu^2}{\varepsilon^2}-\zeta^2} = \frac{\nu}{\varepsilon}\left(1 -\sqrt{ 1 - \frac{\zeta^2\varepsilon^2}{\nu^2}  }\right)
\end{equation*}
and subsequently utilize the inequality
\begin{equation*}
\sqrt{1-t} \leq 1 - \frac{t}{2}.
\end{equation*}
Thus, we obtain the following linear lower bound for $\beta$:
\begin{equation}\label{c2lowerbound}
    \beta \geq \frac{\zeta^2 \varepsilon}{2\nu}.
\end{equation}
To estimate the right-hand side further, we derive a lower bound for $\zeta$. For that, we first write the orthogonal decomposition 
\begin{equation}\label{FrobNormofZ}
    1 = \frob{Z}^2 = \frob{\PT{X_0} Z}^2 + \frob{\PTperp{X_0} Z}^2.
\end{equation}
Returning to the rank-2 decomposition \eqref{blockmatrixdecomposition}, we note that since
\begin{equation*}
     \PT{X_0} Z = 
    \begin{pmatrix}
    h_0   & h_0^{\perp} 
\end{pmatrix}
    \begin{pmatrix}
- \beta & \gamma \\
 \eta & 0 
\end{pmatrix}
\begin{pmatrix}
m^*_0 \\
 {m_0^{\perp}}^*
\end{pmatrix},
\end{equation*}
it holds that
\begin{equation*}
    \frob{\PT{X_0} Z}^2 = \frob{\begin{pmatrix}
- \beta & \gamma \\
 \eta & 0 
\end{pmatrix}}^2 = \beta^2 + \gamma^2 + \eta^2.
\end{equation*}
For the second term in \eqref{FrobNormofZ}, we have (utilizing Lemma \ref{MandBeta})
\begin{equation*}
    \frob{\PTperp{X_0} Z} = \frob{M} \leq \nucl{M} \leq \beta.
    \end{equation*}
In total, we obtain the following chain of inequalities:
\begin{align*}
    \zeta^2 &= (|\gamma| + |\eta|)^2 \geq \gamma^2 + \eta^2 = \frob{Z}^2 - \beta^2 - \frob{M}^2 \\
    &= 1 - \beta^2 - \frob{M}^2 \geq 1 - 2\beta^2.
\end{align*}
Inserting this lower bound in \eqref{c2lowerbound}, we obtain the condition
\begin{equation*}
    \frac{\varepsilon}{\nu}\beta^2 + \beta -\frac{\varepsilon}{2\nu} \geq 0.
\end{equation*}
Solving this inequality in $\beta$, we obtain that
\begin{equation*}
    \beta \in \left(-\infty, -\frac{\nu}{2\varepsilon} - \frac{1}{2}\sqrt{2+\frac{\nu^2}{\varepsilon^2}} \right]  \cup \left[ -\frac{\nu}{2\varepsilon} + \frac{1}{2}\sqrt{2+\frac{\nu^2}{\varepsilon^2}} , \infty \right).
\end{equation*}
Since $\beta$ is nonnegative, this is equivalent to the condition 
\begin{equation}\label{concavefunction}
    \beta \geq -\frac{\nu}{2\varepsilon} + \frac{1}{2}\sqrt{2+\frac{\nu^2}{\varepsilon^2}}.
\end{equation}
We define $f(\varepsilon) := -\frac{\nu}{2\varepsilon} + \frac{1}{2}\sqrt{2+\frac{\nu^2}{\varepsilon^2}}$.
As can be verified via calculating the first and the second derivative, $f$ is an increasing concave function in $\varepsilon$ for $\varepsilon > 0$. We observe that $\lim_{\varepsilon \searrow 0} f(\varepsilon) = 0$ and that $f(2\nu) = \frac{1}{2}$. 
\\[5pt]
Whenever $\varepsilon \leq 2 \nu$, the function $f$, continuously augmented by $f(0) \coloneqq 0$, can be bounded from below by its secant $\frac{\varepsilon}{4\nu}$ going through the points $f(0) = 0$ and $f(2\nu) = 1/2$, since $f$ is concave. Whenever $\varepsilon > 2\nu$, the function $f$ can be bounded from below by the constant $1/2$, since $f$ is increasing. In total, we obtain the lower bound
\begin{equation}\label{betafinallowerbound}
    \beta \geq \min \left\{\frac{\varepsilon}{4\nu}, \frac{1}{2} \right\}.
\end{equation}

\section{Outlook}\label{outlook}
In this paper, we have analyzed robustness of blind deconvolution against adversarial noise and derived a noise-level-dependent reconstruction bound which is consistent with existing evidence for instability of blind deconvolution and matrix completion for sufficiently small noise levels \cite{krahmer2020}. We believe that our approach opens a series of interesting questions for related low-rank matrix recovery problems: 
\begin{enumerate}
    \item \textbf{Noise-dependent stability of matrix completion:} The proof of Theorem \ref{mainresult} works analogously for rank-1 matrix completion, a related low-rank matrix recovery problem in which one wants to reconstruct a rank-$r$ matrix $X_0 \in \C^{n_1 \times n_2}$ from $m$ entries sampled randomly with replacement \cite{CandesRechtExact2009,CandesTaoPower,RechtSimpler,gross11}. However, matrix completion is mostly interesting in case that $X_0$ is a general rank-$r$ matrix. \cite{krahmer2020} showed that matrix completion, similarly to blind deconvolution, can be unstable against adversarial noise if the noise level is sufficiently small. Thus, it remains an interesting open question whether recovery guarantees similar to Proposition \ref{mainresult} hold for rank-$r$ matrix completion and how the required sampling complexity as well as the error bound depend on the rank $r$ of the ground truth.
    \item \textbf{Blind demixing:} An important extension of the blind deconvolution problem is blind demixing \cite{strohmer2017demixing, jung2017optimal}, where one's goal is to reconstruct the pairs of signals $\left( (w_i,x_i) \right)_{i=1}^{r}$ from a (noisy) sum of their convolutions $\sum_{i=1}^{r} w_i \circledast x_i + e$. Similarly to blind deconvolution, this problem can be reformulated to yield a low-rank matrix recovery problem, albeit now of rank $r$ block matrix. Existing results \cite{jung2017optimal,strohmer2017demixing} yield similar reconstruction bounds with noise level amplification by a factor $ \sqrt{r \max \{ K;N\}}$. It poses the question, whether, firstly, similar instability behaviour occurs for blind demixing with deterministic noise, and, secondly, whether guarantees of form \eqref{mainresultequation} hold. 
    \item \textbf{Extension to random noise:} A number of related works, e.g., \cite{Chen2020nonconvex,chenBDminimaxoptimal2021}, have explored robust recovery for blind deconvolution and matrix completion in the case that the measurements are corrupted by random instead of adversarial noise, most common settings being Gaussian (i.e., $e_i \sim \mathcal{N}(0, \sigma^2)$) and sub-Gaussian (i.e., $\spec{e_i}_{\psi_2} \leq \sigma$) noise. 
    This noise setting is particularly convenient when dealing with nonconvex methods and allows to also establish near-optimal recovery guarantees for the convex heuristic by first showing that the convex and nonconvex solutions are mostly close \cite{Chen2020nonconvex}. One of the main drawbacks of this approach is the suboptimal dependence of the sampling complexity on the rank $r$. For (sub)-Gaussian noise, the instability results of \cite{krahmer2020} do not apply anymore. However, it would be interesting to know whether also in the case of random noise, robustness of the recovery is impacted by the noise magnitude and whether our analysis could help improve the $r$-dependence, since there is no more reliance on nonconvex algorithms. Furthermore, other noise settings such as, e.g., Poisson noise \cite{BDpoissonnoise2020}, could be investigated.  
    \item \textbf{Extension to nonconvex methods:} Nonconvex methods based on matrix factorization are often preferred for solving low-rank matrix recovery problems due to their computational efficiency. Several papers have established robustness guarantees for non-convex algorithms in the random noise setting \cite{LiStrohmerRapidDBNonconvex2016, HuangHandBDRiemannian2018, MaImplicitRegularization2017, chenBDminimaxoptimal2021}.  
    \\[5pt]
    Chen et al. \cite{chenBDminimaxoptimal2021} establish closeness of solutions of the convex and nonconvex formulations for blind deconvolution in the random noise scenario, and can thus transfer robustness guarantees for nonconvex algorithms to nuclear norm minimization.
 For establishing robustness against adversarial noise, it is interesting to ask whether one could go in the opposite direction and transfer our results for the convex problem to nonconvex algorithms based on matrix factorization.
    \\[5pt]
    Additionally, there is a line of work which studies low-rank matrix recovery via Iteratively Reweighted Least Squares (IRLS) algorithms \cite{fornasierIRLS2011,fazelIRLS2012,kuemmerleIRLS2018}. Their theoretical analysis is often based on the null space property, which is connected to the geometric relationship between the kernel of the measurement operator and the descent cone of the ground truth. However, this property does not hold for more structured measurements such as in the blind deconvolution setting. Thus, existing theoretical guarantees for IRLS dox not apply to blind deconvolution and similar problems. It is interesting to ask whether the geometric insights in this paper could be utilized to analyse IRLS in such settings.
\end{enumerate}
\section*{Acknowledgements}

F.K. and J.K. acknowledge support by the German Ministry of Education and Research (BMBF) in the context of the Munich Center for Machine Learning (MCML) and by the German Science Foundation (DFG) in the context of the project \emph{Solving linear inverse problems with end-to-end neural networks: expressivity, generalization, and robustness} project number 464123524 as part of the  Priority Program 2298.\\
F.K. and D.S. acknowledge support by the German Science Foundation (DFG) in the context of the project \emph{Bilinear Compressed Sensing - Efficiency, Structure, and Robustness} project number 273529854 as a part of the Priority Program 1798.

\printbibliography

\appendix

\section{Construction of the Dual Certificate via the Golfing Scheme}
As we have mentioned in Section \ref{background:DCandtheGolfingScheme}, the exact dual certificate for $X_0$ exists if $X_0$ is a minimizer of the noiseless problem \eqref{noiselessSDP}, which has been established with high probability in \cite{ahmedBD,jung2017optimal,strohmer2017demixing}.
However, the explicit construction of an exact dual certificate (done, e.g., in \cite{CandesRechtExact2009} for matrix completion) is rather tedious, and thus, unique and/or robust recovery for blind deconvolution has been typically shown using so-called \emph{approximate dual certificates} instead \cite{gross11,RechtSimpler}. In the following sections, we will outline the explicit construction of an approximate dual certificate by the means of the Golfing scheme and then show how an \emph{exact} dual certificate can be explicitly constructed from an approximate dual certificate \cite{Fuchs2021}. Furthermore, we will show that both the approximate and the exact dual certificates are bounded in $\ell_2$-norm, which is crucial for our proof of Lemma \ref{lowerboundforAZ}.

\subsection{Approximate Dual Certificate}\label{appendix:dualcertificate}
We first introduce the notion of an approximate dual certificate, i.e., an object which satisfies the properties in Proposition \ref{exactDCexistence} up to a small error:
\begin{definition}[\cite{gross11,RechtSimpler}]\label{approximatedualcertificate}
        Let $z \in \C^L$. Let $X_0 =  h_0 \nu  m_0^*$ be the singular value decomposition of the rank-1 matrix $X_0$. $Y \coloneqq \As(z)$ is called an \emph{approximate dual certificate} if it satisfies
    \begin{align}
        \frob{\PT{X_0}Y - h_0 m_0^*} \, &\leq \frac{1}{8 \spec{\A}},\label{property1} \\
        \spec{\PTperp{X_0} Y} \, &< \frac{1}{2}.\label{property2}
    \end{align}
\end{definition}
For the operator norm of $\A$, the following upper bound holds with high probability.
\begin{lemma}[\cite{ahmedBD,jung2017optimal}]\label{lemmaNormofA}
Let $\omega \geq 1$. Then with probability at least $1-2 L^{-\omega}$
\begin{equation}\label{Aoperatornormbound}
    \spec{\A} \, \leq 2 \sqrt{ \omega \max \left\{1;\frac{\mu_{\max} K N}{L}\right\} \log(L + KN)} =: \tilde{\zeta}.
\end{equation}
\end{lemma}

\subsection{The Golfing Scheme}\label{appendix:golfingscheme}
In this section, we will outline the Golfing scheme \cite{gross11,ahmedBD}.
Our presentation will be based on \cite{strohmer2017demixing,jung2017optimal}, which analyze the Golfing Scheme for the more general scenario of blind demixing.
\\[5pt]
\textbf{1. Existence of an admissible partition.} 
The first step in the Golfing scheme is to find a partition $\{\Gamma_p\}_{p=1}^{P}$ of the set of the measurements $[L]$ into $P$ smaller sets and to construct the associated projected operators $\A_p, \, p=1,...,P$, where we define $\A_p \coloneqq \mathcal{P}_{\Gamma_p} \A$. We will denote by $Q = \frac{L}{P}$ the approximate number of measurement in each partition. One of the requirements for the success of the Golfing scheme is that the random operators $\frac{L}{Q} \A_p^* \A_p$ act proportionate to an (approximate) identity in expectation, which translates into the requirement
\begin{equation*}
    T_p \coloneqq \frac{L}{Q} \sum_{k \in \Gamma_p} b_k b_k^* \approx Id_{K},
\end{equation*}
where $Id_{K}$ denotes the identity operator on $\C^K$.  
In general, one can only assure that $\max_{p \in [P]} \spec{Id_K - T_p} \leq \alpha$, for $\alpha > 0$. 
Together with further requirements on the number $P$ and size $Q$ of the individual subsets $\Gamma_p$, this results in the following notion of an $\omega$-admissible partition.
\begin{definition}[\cite{jung2017optimal}]\label{admissiblepartition}
  Let $\omega \geq 1$ and $\{\Gamma_p\}_{p=1}^P$ such that $[L] = \bigcup_{p=1}^P \Gamma_p$. Then, $\{\Gamma_p\}_{p=1}^P$ is called an $\omega$-\emph{admissible partition} if the following conditions hold:
  \begin{enumerate}
      \item $\frac{Q}{2} \leq |\Gamma_p| \leq \frac{3Q}{2}$ for all $p \in [P]$;
      \item For $1 \leq p \leq P$, $\spec{Id_K - T_p}\, \leq \alpha$ for some $\alpha \leq \frac{1}{32}$;
      \item $\log(8 \tilde{\zeta}) \geq P \geq \frac{1}{2} \log{8 \tilde{\zeta}}$,
  \end{enumerate}
  where $\tilde{\zeta} = 2 \sqrt{ \omega \max \left\{1;\frac{\mu_{\max} K N}{L}\right\} \log(L + KN)}$.
\end{definition}
\cite{jung2017optimal} and \cite{strohmer2017demixing} show the existence of such a partition for the more general setting of blind demixing. Here, we present their result applied to the blind deconvolution scenario.
\begin{lemma}[\cite{jung2017optimal,strohmer2017demixing}]
   Let $P \in [L]$ and $\alpha \in (0,1)$. If the number of measurements satisfies 
   \begin{equation}
       L \gtrsim \frac{1}{\alpha^2} \log(8 \tilde{\zeta}) \mu_{\max} K \log(\max\{P;K\}),
   \end{equation}
then there exists a partition  $\{\Gamma_p\}_{p=1}^P$ of $L$ which satisfies conditions 1. and 2. in Definition \ref{admissiblepartition}. 
\end{lemma}
In the following, we will choose an $\omega$-admissible partition which minimizes the coherence of the input $h_0$ with respect to the rows $b_{\ell}$ distorted by a set of linear maps related to the partition $\{\Gamma_p\}_{p=1}^{P}$. This partition will define the \emph{minimal coherence parameter} by
\begin{equation}\label{muhomegadefinition}
    \mu_{h_0, \omega}^2 \coloneqq L \min_{\{\Gamma_p\}_{p=1}^{P} \omega-\text{adm.}} \max\left\{ \max_{\ell \in [L]} |b^*_{\ell} h_0|^2, \: \max_{l \in [L], p \in [P]} |b_{\ell}^* S_p h_0|^2 \right\},
\end{equation}
where we have defined $S_p \coloneqq T_p^{-1}$ (note that by Definition \ref{admissiblepartition}, $T_p$ is invertible).
\\[5pt]
\textbf{2. The Golfing Scheme.}
We will now outline the random process by the means of which an approximate dual certificate $Y \in \text{Range}(\As)$ satisfying Definition \ref{approximatedualcertificate} is constructed in \cite{jung2017optimal}. We set
\begin{align*}
    Y_0 &= 0; \\
    Y_p &= Y_{p-1} + \frac{L}{Q} \A_p^* \A_p S_p (h_0 m_0^* - \PT{} Y_{p-1}) \text{ for } p \in [P],
\end{align*}
where $S_p$ is used as a corrector function to ensure that $\mathbb{E}\left[ \frac{L}{Q} \A_p^* \A_p S_p X \right] = X$ for all matrices $ X \in \C^{K \times N} $, since 
\begin{equation*}
    \mathbb{E}\left[ \frac{L}{Q} \A_p^* \A_p (X) \right] = T_p X.
\end{equation*}
Using the short notation $W_p \coloneqq h_0 m_0^* - \PT{} Y_p$, we define
\begin{equation}\label{golfingschemedefinition}
    Y \coloneqq Y_P = \sum_{p=1}^P \frac{L}{Q} \A_p^* \A_p S_p (W_{p-1}).
\end{equation}
Next, we check that $Y$ is indeed in the range of $\As$. For this, define
\begin{equation}\label{zdefinition}
    z = \sum_{p=1}^P \frac{L}{Q} \A_p S_p (W_{p-1}).
    \end{equation}
We recall that since $\Gamma_p$ are disjoint subsets and $\A_p$ sets all components not belonging to $\Gamma_p$ to zero, it actually holds that
\begin{equation}\label{ApisA}
    \A_p^* \A_p S_p (W_{p-1}) = \A^* \A_p S_p (W_{p-1})
\end{equation}
and thus $Y = \As(z)$.
\\[5pt]
\textbf{3. The $\delta$-restricted isometry property.}
To establish that $Y$ indeed satisfies Definition \ref{approximatedualcertificate}, one first needs to show that $\A$ acts as an approximate isometry on the tangent space $\tangent{X_0}$. This so-called \emph{restricted isometry property (RIP)} has been long utilized in the field of sparse recovery to prove exact or stable signal reconstruction \cite{CandesTaoDecoding2005, CandesTaoRIP2006}. 
\begin{definition}[$\delta$-restricted isometry property]\label{RIP-definition}
   Let $\delta > 0$. A linear operator $\A: \C^{K \times N} \to \C^L$ is said to satisfy the $\delta$-restricted isometry property (RIP) on a subset $\mathcal{M} \subset \C^{K \times N}$ if for all $Z \in \mathcal{M}$
   \begin{equation*}
       (1-\delta)\frob{Z}^2 \leq \twonorm{\A(Z)}^2 \leq (1+\delta) \frob{Z}^2.
   \end{equation*}
\end{definition}
Whereas Gaussian measurement operators generally satisfy the RIP on the whole domain with high probability, for more structured measurements, like in the case of blind deconvolution or matrix completion, the restricted isometry property can only be established on small subspaces.
\\[5pt]
The following proposition states that the measurement operator $\A$ satisfies the $\delta$-RIP on $\tangent{X_0}=\tangent{}$ with high probability provided that the number of observations $L$ scales at least inverse quadratically in $\delta$.
Moreover, under the same conditions, the partial operators $\A_p = \mathcal{P}_{\Gamma_p} \A$ satisfy the $\delta$-RIP on slightly larger subspaces $\tangent{}^p \coloneqq \tangent{} + S_p(\tangent{\, })$.
\begin{proposition}[\cite{jung2017optimal}]\label{RIPdef}
    Let $\delta > 0$ and fix $\omega \geq 1$. Let $\A: \C^{K \times N} \to \C^L$ as in \eqref{LinearOperator}. Let 
    \begin{equation}\label{QforRIP}
       Q \geq C \omega \delta^{-2} ( K \mu_{\max} \log{L} \log^2(K \mu_{\max}) + N \mu_{h_0}^2).
    \end{equation}
    Then with probability at least $1 - \mathcal{O}(L^{-\omega})$, all $X \in \tangent{}$ satisfy
    \begin{equation}\label{RIP1}
        (1 - \delta) \frob{X}^2 \, \leq \, \twonorm{\A(X)}^2  \leq \, (1 + \delta) \frob{X}^2,
    \end{equation}
    and for all $p \in [P]$, all $Y \in \tangent{}^p$ satisfy
    \begin{equation}\label{RIP2}
        (1 - \delta) \frob{T^{1/2}_p Y}^2 \leq \frac{L}{Q} \twonorm{\A_p (Y)}^2 \leq (1+ \delta) \frob{T^{1/2}_p Y}^2.
    \end{equation}
\end{proposition} \textbf{4. Dual certificate properties.} To establish that the constructed matrix $Y$ indeed satisfies Definition \ref{approximatedualcertificate} with high probability, \cite{jung2017optimal} first establish the following exponential decay behaviour: 
\begin{lemma}[\cite{jung2017optimal}]\label{lemmadecaybehaviour}
    Assume that for all $p \in [P]$, the partial operators $\A^p$ satisfy the $\delta$-RIP on $\tangent{}^p$ with $\delta = \frac{1}{32}$. Then, for all $p \in [P]$,
    \begin{equation*}
        \frob{W_p} \leq 4^{-p}.
    \end{equation*}
    Now, set the number of subsets in the partition to $P = \frac{1}{2} \log(8 \spec{\A})$. Then
    \begin{equation*}
        \frob{h_0 m_0^* - Y} \leq \frac{1}{8 \spec{\A}},
    \end{equation*}
i.e., $Y$ satisfies property \eqref{property1}. 
\end{lemma}
The following lemma establishes an upper bound for the spectral norm $\spec{\PTperp{} Y}$:
\begin{lemma}[\cite{jung2017optimal}]\label{PTperpYlemma}
    Under the assumptions of Lemma \ref{lemmadecaybehaviour} and if 
    \begin{equation*}
        Q \geq C \omega (K \mu_{\max} + N \mu^2_{h_0}) (\log L)^2,
    \end{equation*}
 with probability at least $1-\mathcal{O}(L^{-\omega})$ it holds that
    \begin{equation*}
        \spec{\PTperp{} Y_P} < \frac{1}{4},
    \end{equation*}
    where $C$ is a universal constant.
\end{lemma}
\textbf{5. Norm upper bound for the dual certificate.} We can now show that $z$ defined in \eqref{zdefinition} is bounded with respect to the $\ell_2$-norm. The following lemma is based on Lemma 5.18 in \cite{jung2017optimal}:\footnote{We thank Dana Weitzner for pointing out an inaccuracy in the original proof in \cite{jung2017optimal}. The corrected proof introduces an additional logarithmic factor to the upper bound of $\twonorm{z}$.}
\begin{lemma}\label{lemmazbound}
    Let $z \in \C^L$ be given by \eqref{zdefinition}. Assume that the measurement operator $\A$ satisfies the operator norm bound \eqref{Aoperatornormbound} for some $\omega \geq 1$. Under the assumptions of Lemma \ref{lemmadecaybehaviour}, it holds that
\begin{equation*}
    \twonorm{z} \lesssim  \sqrt{\log(\omega L)}.
\end{equation*}
\end{lemma}
\begin{proof}
Since $z$ is given by 
\begin{equation*}
    z = \sum_{p=1}^P \frac{L}{Q} \A_p S_p (W_{p-1}),
\end{equation*}
it holds that
\begin{equation}\label{ztriangleinequality}
    \twonorm{z} \leq \frac{L}{Q} \sum_{p=1}^P \twonorm{\A_p S_p (W_{p-1})} = P  \sum_{p=1}^P \twonorm{\A_p S_p (W_{p-1})}.
\end{equation}
Utilizing the isometry property \eqref{RIP2}, we observe that
\begin{align*}
    \twonorm{\A_p S_p (W_{p-1})} \lesssim \sqrt{\frac{Q}{L}} \frob{T^{1/2}_p S_p (W_{p-1})} \lesssim \sqrt{\frac{Q}{L}} \frob{S_p (W_{p-1})} \lesssim \sqrt{\frac{Q}{L}} \frob{W_{p-1}},
\end{align*}
where we have also used that $\spec{Id - T^{1/2}_p} \leq \frac{1}{32}$ and $\frob{S_p (W_{p-1})} \leq \frac{32}{31} \frob{W_{p-1}}$, see Lemma 5.13 in \cite{jung2017optimal}. Inserting this into \eqref{ztriangleinequality}, we obtain
\begin{align*}
    \twonorm{z} \lesssim \sqrt{\frac{L}{Q}} \sum_{p=1}^P \twonorm{W_{p-1}}. 
\end{align*}
Together with the decay property of $\frob{W_p}$ from Lemma \ref{lemmadecaybehaviour} this implies
\begin{equation*}
    \twonorm{z} \lesssim \sqrt{\frac{L}{Q}} \sum_{p=1}^P 4^{-p} \lesssim  \sqrt{\frac{L}{Q}} = \sqrt{P}.
\end{equation*}
Recall that our choice for the size of the partition was $P = \frac{1}{2} \log(8 \spec{\A})$ and according to Lemma \ref{lemmaNormofA}, with probability at least $1 - 2 L^{-\omega}$ it holds that
\begin{equation*}
    \spec{\A} \, \leq 2 \sqrt{ \omega \max \left\{1;\frac{\mu_{\max} K N}{L}\right\} \log(L + KN)}.
\end{equation*}
Therefore, with probability at least $1 - 2 L^{-\omega}$ 
\begin{align*}
    \twonorm{z} \lesssim \sqrt{ \log(\spec{\A})} \lesssim \sqrt{\log(\omega L \log L)} \lesssim \sqrt{\log(\omega L)}.
\end{align*}
\end{proof}

\subsection{From approximate to exact dual certificate}\label{appendix:puttingproposition}
The following proposition states that constructing an approximate dual certificate which satisfies Definition \ref{approximatedualcertificate} via the Golfing scheme automatically yields an exact dual certificate. Furthermore, we will see that if the original approximate dual certificate is bounded with respect to the $\ell_2$-norm, the same asymptotic bound also holds for the exact dual certificate derived from it. We present the result following \cite{Fuchs2021} and add a derivation of the norm bound. 

\begin{proposition}[\cite{Fuchs2021}]\label{puttingproposition}
    Assume that there exists an approximate dual certificate pair $(z, Y)$, where $Y = \As(z)$, satisfying properties \ref{property1} and \ref{property2}, and assume that $\A$ satisfies the $\delta$-restricted isometry property on $\tangent{X_0}$ for $\delta < 3/4$. Then there exists an exact dual certificate pair $(z',Y')$ with $Y' = 
    \A(z')$ satisfying properties \eqref{property1exact} and \eqref{property2exact}. Furthermore, if it holds that 
    \begin{equation*}
        \twonorm{z'} \leq \twonorm{z} + 1.
    \end{equation*}
\end{proposition}
\begin{proof}

    First, we observe that the $\delta$-restricted isometry property is equivalent to the fact that 
    
\begin{align*}
        \spec{\PT{X_0} \As \A \PT{X_0} - \PT{X_0}} &= \spec{ \PT{X_0} ( \As \A - Id ) \PT{X_0} } \\
        &= \sup_{Z \in \tangent{X_0}, \frob{Z}=1} \left| \scal{Z, \As \A (Z)} - 1 \right| \\
        &=  \sup_{Z \in \tangent{X_0}, \frob{Z}=1} \left| \frob{\A (Z)}^2 - 1 \right|  \leq \delta.
\end{align*}
Thus, $\PT{X_0} \As \A \PT{X_0}: \tangent{X_0} \to \tangent{X_0}$ is invertible (e.g., via the Neumann series) and satisfies
\begin{equation}\label{normofpaap}
    \spec{ (\PT{X_0} \As \A \PT{X_0})^{-1}} \leq \frac{1}{1-\delta}.
\end{equation}
We define
\begin{equation*}
    x = \A \PT{X_0} (\PT{X_0} \As \A \PT{X_0})^{-1} (h_0 m_0^*- \PT{X_0} \As(z))
\end{equation*}
and 
\begin{align*}
    z' &= z + x, \\
    Y' &= \As(z') = Y + \As(x).
\end{align*}
First, we observe that 
\begin{align*}
    \spec{ \A \PT{X_0} (\PT{X_0} \As \A \PT{X_0})^{-1}  }^2 = \sup_{X \in \C^{K \times N} } \twonorm{\A \PT{X_0} (\PT{X_0} \As \A \PT{X_0})^{-1} X}^2 \\
    = \sup_{X \in \C^{K \times N} } \scal{ \A \PT{X_0} (\PT{X_0} \As \A \PT{X_0})^{-1} X, \, \A \PT{X_0} (\PT{X_0} \As \A \PT{X_0})^{-1} X } \\
    = \sup_{X \in \C^{K \times N} } \scal{ X, (\PT{X_0} \As \A \PT{X_0})^{-1} X } \leq \frac{1}{1-\delta}.
\end{align*}
Thus, 
\begin{equation*}
     \spec{ \A \PT{X_0} (\PT{X_0} \As \A \PT{X_0})^{-1}  } \leq \frac{1}{\sqrt{1-\delta}},
\end{equation*}
and
\begin{equation*}
    \twonorm{x} \leq \frac{1}{\sqrt{1-\delta}} \frob{h_0 m_0^* - \PT{X_0} \As(z)} \leq \frac{1}{8 \spec{\A} \sqrt{1-\delta} }
\end{equation*}
using property \ref{property1} of the approximate dual certificate. Furthermore, we obtain
\begin{equation*}
    \PT{X_0} Y' = h_0 m_0^*,
\end{equation*}
since
\begin{align*}
    \PT{X_0} Y' &= \PT{X_0} Y + \PT{X_0}\As(x) \\
    &= \PT{X_0}\As(z)+\PT{X_0}\As\A\PT{X_0}(\PT{X_0}\As\A\PT{X_0})^{-1}(h_0 m_0^* - \PT{X_0} \As(z)) \\
    & = \PT{X_0}\As(z) + h_0 m_0^* - \PT{X_0} \As(z)) = h_0 m_0^*.
\end{align*}
Hence, $Y'$ fulfills property \eqref{property1exact}. Furthermore,
\begin{align*}
    \spec{\PTperp{X_0}Y' } &\leq \spec{\PTperp{X_0} Y   } + \spec{ \PTperp{X_0} \As (x)  } \nonumber \\
    &\leq \frac{1}{2} + \frob{ \PTperp{X_0} \As (x)  } \nonumber \\
    &\leq \frac{1}{2} + \spec{\A} \twonorm{x} \nonumber \\
    &\leq  \frac{1}{2} + \frac{1}{8 \sqrt{1-\delta}}  < 1, \label{PTperpwith34}
\end{align*}
i.e., $Y'$ fulfills property \eqref{property2exact}. Next, we estimate the $\ell_2$-norm of $z'$ via
\begin{align*}
    \twonorm{z'} = \twonorm{z + x} \leq \twonorm{z} + \twonorm{x} \leq \twonorm{z} + \frac{1}{8 \sqrt{1-\delta} \spec{\A}}.
\end{align*}
Now, since we have assumed that $\A$ fullfills the $\delta$-RIP on $\tangent{}$ with $\delta < 3/4$, it holds that 
\begin{equation*}
    \spec{\A} \geq 
\sqrt{1-\delta} \geq 1/2.
\end{equation*}
Therefore, we conclude
\begin{equation*}
    \twonorm{z'} \leq  \twonorm{z}  + \frac{1}{8 (1-\delta)} \leq \twonorm{z}  + 1.
\end{equation*}
\end{proof}
\subsection{Proof of Proposition \ref{exactDCexistence}}

Now we have all ingredients in place to give a proof of Proposition \ref{exactDCexistence}.

\begin{proof}[Proof of Proposition \ref{exactDCexistence}]
Assume that the number of observations $L$ satisfies \eqref{LnumforexistenceofDC}. We choose an $\omega$-admissible partition $\{ \Gamma_p \}_{p = 1}^{P}$ which minimizes the minimal coherence parameter \eqref{muhomegadefinition}. We construct $Y$ according to \eqref{golfingschemedefinition} and $z$ according to \eqref{zdefinition}. Then it holds that $Y = \As(z)$ (see \eqref{ApisA}). 
\\[5pt]
For now, we assume that $\A$ satisfies $\delta$-RIP with $\delta \leq \frac{1}{32}$ on $\tangent{X_0}$, operators $\A_p$ satisfy the $\delta$-RIP with $\delta \leq \frac{1}{32}$ on spaces $\tangent{}^p$ for all $p \in [P]$, and that the bound on the operator norm of $\A$ as in Lemma \ref{lemmaNormofA} holds. Furthermore, we assume that the conclusion of Lemma \ref{PTperpYlemma} holds. Then, according to Lemma \ref{lemmadecaybehaviour} and Lemma \ref{PTperpYlemma}, $Y$ satisfies both approximate dual certificate properties \eqref{property1} and \eqref{property2}, and thus $(z,Y)$ constitute a dual certificate pair. Moreover, according to Lemma \ref{lemmazbound}, $z$ satisfies 
\begin{equation*}
    \twonorm{z} \lesssim \sqrt{\log(\omega L) }.
\end{equation*}
Finally, according to Proposition \ref{puttingproposition}, there exists an exact dual certificate pair $(z', Y')$ for which it holds that $\twonorm{z'} \leq \twonorm{z} + 1$. Thus, it follows that the norm of the exact dual certificate is also bounded:
\begin{equation*}
    \twonorm{z'} \lesssim \sqrt{\log(\omega L) }.
\end{equation*}
We have conditioned on the following events: $\delta$-RIP for the operators $\A$ and $\A_p$, $p \in [P]$, the operator norm bound for $\A$, and the bound on the spectral norm of $\PTperp{} Y$. We now choose the absolute constant in the number of observations large enough such that $\delta$-RIP for $\A$ and $\A_p$, $p \in [P]$ with $\delta \leq \frac{1}{32}$ holds with probability at least $1-\mathcal{O}(L^{-\omega})$ according to Proposition \ref{RIPdef}, and such that the upper bound for the operator norm of $\A$ holds with probability at least $1-2 L^{-\omega}$ according to Lemma \ref{lemmaNormofA}. Moreover, according to Lemma \ref{PTperpYlemma}, the upper bound on the spectral norm of $\PTperp{} Y$ holds with probability at least $1 - \mathcal{O}(L^{-\omega})$. Adjusting the constant in the number of observations $L$ and taking a union bound then finalizes the proof. 
\end{proof}

\end{document}